%% file: linecongruences.tex
\documentclass{amsart}
\usepackage{amsmath,amssymb,amsfonts}
\usepackage{graphicx}
\usepackage[alphabetic,y2k]{amsrefs}
\usepackage{xcolor}
\usepackage{pst-all}
\usepackage{pst-solides3d}

\def\Gr {\mathrm{Gr}}

\newtheorem{theorem}{Theorem}[section]
\newtheorem{proposition}[theorem]{Proposition}
\newtheorem{corollary}[theorem]{Corollary}
\newtheorem{lemma}[theorem]{Lemma}
\newtheorem{remark}[theorem]{Remark}
\newtheorem{example}[theorem]{Example}
\newtheorem{definition}[theorem]{Definition}

\date{}

\begin{document}

\input{hyphenation}

\author{Vladimir Dragovi\'c}
\address{Mathematical Institute SANU, Kneza Mihaila 36, Belgrade,
Serbia\newline\indent Mathematical Physics Group, University of
Lisbon, Portugal} \email{vladad@mi.sanu.ac.rs}

\author{Milena Radnovi\'c}
\address{Mathematical Institute SANU, Kneza Mihaila 36, Belgrade, Serbia}
\email{milena@mi.sanu.ac.rs}

\title[Billiards, line congruences, and double reflection nets]{Billiard algebra, integrable line congruences, and double reflection nets}

\begin{abstract}
The billiard systems within quadrics, playing the role of  discrete analogues of geodesics on ellipsoids, are incorporated into the theory of integrable quad-graphs.
An initial observation is that the Six-pointed star theorem, as the operational consistency for the billiard algebra, is equivalent to an integrabilty condition of a line congruence.
A new notion of the double-reflection nets as a subclass of dual Darboux nets associated with pencils of quadrics is introduced, basic properies and several examples are presented.
Corresponding  Yang-Baxter maps, associated with pencils of quadrics are defined and discussed.
\end{abstract}

\maketitle



\setcounter{tocdepth}{1}%
\tableofcontents

\setlength{\parskip}{2pt}

\section{Introduction}
\label{sec:intro}
\input{1-intro/intro}

\section{Line congruences and quad-graphs integrability}
\label{sec:line}
\input{2-congruences/cong}

\section{Double reflection configuration}\label{sec:drc}
\input{3-DRC/DRC}

\section{Billiard algebra, construction of quad-graphs, and 3D-consistency}
\label{sec:billiard}
\input{4-algebra/billiard-algebra1}

\section{Double reflection nets}
\label{sec:drnets}
\input{5-DRnets/drnets}

\section{Yang-Baxter map}
\label{sec:yb}
\input{6-YB/yb}

\subsection*{Acknowledgements}
The authors are grateful to Prof.~Yu.~Suris for inspirative discussions, to Dr.~Emanuel Huhnen-Venedey for helpful remarks, and to the referee whose comments led to significant improvement of the manuscript.
The research was partially supported by the Serbian Ministry of Science and Technology, Project 174020 \emph{Geometry and Topology of Manifolds, Classical Mechanics and Integrable Dynamical Systems} and by the Mathematical Physics Group of the University of Lisbon, Project \emph{Probabilistic approach to finite and infinite
dimensional dynamical systems, PTDC/MAT/104173/2008}.
M.~R.\ is grateful to the Abdus Salam ICTP from Trieste and its associateship programme for support.

\end{document}

%% file: hyphenation.tex
\hyphenation{Abe-li-an}

\hyphenation{boun-da-ry}

\hyphenation{Cha-o-tic cur-ves}

\hyphenation{El-lip-ti-cal en-coun-ter}

\hyphenation{Lo-ba-chev-sky}

\hyphenation{Mar-den Min-kow-ski}

\hyphenation{pa-ra-met-ri-zes pa-ra-met-ri-za-tion
Pon-ce-let-Dar-boux}

\hyphenation{quad-ric quad-rics qua-dri-ques}

\hyphenation{tra-jec-to-ry trans-ver-sal}

%% file: 1-intro/intro.tex
Although its historic roots may be traced back to Newton’s \emph{Principia},
the Discrete Differential Geometry emerged as a modern scientific discipline quite recently (see \cite{BS2008book}),
within a study of lattice geometry, where so-called integrability or consistency conditions for quad graphs have been playing a fundamental role.

Geodesics on an ellipsoid present one of the most important and exciting basic examples of the Classical Differential Geometry.
Billiard systems within quadrics (see \cite{DragRadn2011book} for a detailed study),
being discretization of systems of geodesics on ellipsoids,
should be seen as an important part of the foundations of the Discrete Differential Geometry.

The main aim of the present paper is to incorporate billiard systems within quadrics into the story of integrable quad-graphs.
In Section \ref{sec:line}, necessary notions of the theory of integrable systems on quad-graphs from works of Adler, Bobenko, Suris \cite{ABS2003, ABS2009b, BS2008book} are reviewed.
In Section \ref{sec:drc}, we recall definition and main properties of Double reflection configuration from \cite{DragRadn2006,DragRadn2008} and in Proposition \ref{prop:drc.quad} show that this configuration can play a geometric counterpart role of the quad-equation.
In Section \ref{sec:billiard}, we use the billiard algebra, which is developed in our paper \cite{DragRadn2008} as an algebraic tool for synthetic realization of higher-genera addition theorems.
Our main observation there is that the Six-pointed star theorem, i.e.~Theorem \ref{th:zvezda} of the present paper, on a consistency of a certain projective configuration, as the operational consistency for the billiard algebra, is equivalent to an integrabilty condition of a line congruence.
In Section \ref{sec:drnets}, a new notion of double reflection nets is introduced.
We provide four illustrative examples of double reflection nets: two are based on the Poncelet-Darboux grids from \cite{DragRadn2008}, third one on our study of ellipsoidal billiards in pseudo-Euclidean spaces (see \cite{DragRadn2011}), and in the last one we give a general construction of a double reflection net.
After this, $F$-transformations and focal nets for double reflection nets are discussed.
We show that double reflection nets induce a subclass of Grassmannian Darboux nets from \cite{ABS2009} associated with pencils of quadrics.
As we know after Schief (see \cite{Schief2003}) the Darboux nets are associated
to discrete integrable hierarchies.
In Section \ref{sec:yb}, we conclude by definition and study of corresponding Yang-Baxter maps, associated with pencils of quadrics, leading to a natural generalization of the situation analyzed in \cite{ABS2004}.

%% file: 2-congruences/cong.tex
Now, we will start with basic ideas of the theory of integrable systems on quad-graphs from works of Adler, Bobenko, Suris \cite{ABS2003, ABS2009b, BS2008book}.

Recall that the basic building blocks of systems on quad-graphs are
the equations on quadrilaterals of the form
\begin{equation}\label{eq:quadeq}
Q(x_1,x_2,x_3,x_4)=0,
\end{equation}
where $Q$ is a multiaffine polynomial, that is a polynomial of degree one in each argument.

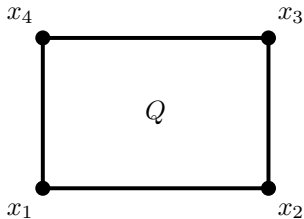
\begin{figure}[h]
\centering
\input{0-figures/quadrilateral}
 \caption{Elementary quadrilateral; quad-equation $Q(x_1,x_2,x_3,x_4)=0$.}\label{fig:quadrilateral}
\end{figure}

Equations of type (\ref{eq:quadeq}) are called \textit{quad-equations}.
The field variables $x_i$ are assigned to four vertices of a quadrilateral as in Figure \ref{fig:quadrilateral}.
Equation (\ref{eq:quadeq}) can be solved for each variable, and the solution is a rational function of the other three variables.
A solution $(a_1,a_2,a_3,a_4)$ of equation (\ref{eq:quadeq}) is \textit{singular} with respect to $x_i$ if it also satisfies the equation:
$$
\frac{\partial Q}{\partial x_i}(a_1,a_2,a_3,a_4)=0.
$$

Following \cite{ABS2009b}, we consider the idea of integrability as
consistency, see Figure \ref{fig:cube}.
\begin{figure}[h]
\centering
\input{0-figures/cube}
 \caption{3D-consistency.}\label{fig:cube}
\end{figure}
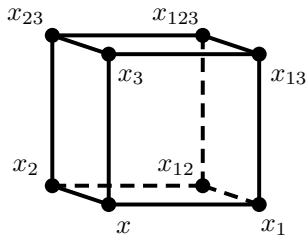
We assign six quad-equations to the faces of coordinate cube:
\begin{gather*}
Q(x,x_1,x_{12},x_2)=0, \quad
Q(x,x_1,x_{13},x_3)=0, \quad
Q(x,x_2,x_{23},x_3)=0,\\
Q(x_1,x_{12},x_{123},x_{13})=0,\quad
Q(x_2,x_{12},x_{123},x_{23})=0,\quad
Q(x_3,x_{13},x_{123},x_{23})=0.
\end{gather*}
Notice that, for given $x$, $x_1$, $x_2$, $x_3$, values of $x_{12}$, $x_{13}$, $x_{23}$ can be determined from the first three of these quad-equations.
In general, the three values for $x_{123}$ determined from the rest three quad-equations may be distinct. 

The system is said to be \textit{3D-consistent} if the three values for $x_{123}$ obtained from these equations coincide for arbitrary initial data 
$x$, $x_1$, $x_2$, $x_3$.

We will be interested in geometric version of integrable quad-graphs, with lines in
$\mathbf P^d$ playing a role of vertex fields.
More precisely, denote by $\mathcal{L}^d$ the Grassmannian $\Gr(2,d+1)$ of two-dimensional vector subspaces of the $(d+1)$-dimensional vector space, $d\ge2$.

For a map
$$
\ell:\mathbf{Z}^m \rightarrow\mathcal{L}^d,
$$
we denote: $\ell_i:\mathbf{Z}^m \rightarrow\mathcal{L}^d$, such that
$\ell_i(u)=\ell(u+\mathbf{e}_i)$ for each $u\in\mathbf{Z}^m$, where
$\mathbf{e}_i\in\mathbf{Z}^m$ is the $i$-th unit coordinate vector, $i\in\{1,\dots,m\}$.

Following \cite{BS2008book}, we say that $\ell$ is \emph{a line
congruence} if neighbouring lines $\ell(u)$ and $\ell_i(u)$
intersect for each $u\in\mathbf{Z}^m$ and $i\in\{1,\dots,m\}$.

Given a line congruence $\ell: \mathbf Z^m \rightarrow \mathcal{L}^d$, one may define its \emph{$i$-th focal net} as the map
$$
F^{(i)}:\mathbf Z^m\rightarrow \mathbf P^{d+1},
\quad
F^{(i)}(u)=\ell(u)\cap\ell_i(u).
$$

For two given line congruences
$\ell, \ell^+: \mathbf{Z}^m\rightarrow \mathcal{L}^d $,
there is an \emph{$F$-transformation} if for each $u\in \mathbf Z^m $ the corresponding lines $\ell(u)$
and $\ell^+(u)$ intersect.

As we mentioned before, in the case of algebraic quad-graphs, there is a multiaffine function $Q$, which determines a fourth vertex of a quadrilateral, according to the other three.
In the sequel, we are going to introduce a geometric configuration of double reflection, which is going to play a role of a multiaffine function $Q$ in the case of line congruences we are going to study.
In other words, given three lines in a configuration, the fourth line will be uniquely determined, as it will be proved in Proposition \ref{prop:drc.quad} at the end of Section \ref{sec:drc}.
Moreover, using properties of billiard algebra developed in \cite{DragRadn2008}, we will prove that such a quad-relation is $3D$-consistent, see Theorem \ref{th:3d}.

%% file: 0-figures/quadrilateral.tex
\begin{pspicture}(-0.5,-0.5)(3.5,2.5)
\psset{
    fillstyle=none, linecolor=black, linestyle=solid, linewidth=1.5pt
    }

\pscircle*(0,0){0.1}
\pscircle*(0,2){0.1}
\pscircle*(3,2){0.1}
\pscircle*(3,0){0.1}

\psline(0,0)(0,2)(3,2)(3,0)(0,0)

\rput(-0.3,-0.3){$x_1$}
\rput(-0.3,2.3){$x_4$}
\rput(3.3,2.3){$x_3$}
\rput(3.3,-0.3){$x_2$}

\rput(1.5,1){$Q$}
\end{pspicture}

%% file: 0-figures/cube.tex
\begin{pspicture}(-1.25,-0.5)(2.5,2.75)
\psset{
    fillstyle=none, linecolor=black, linestyle=solid, linewidth=1.5pt
    }

\pscircle*(0,0){0.1}
\pscircle*(0,2){0.1}
\pscircle*(2,2){0.1}
\pscircle*(2,0){0.1}

\pscircle*(-0.75,0.25){0.1}
\pscircle*(-0.75,2.25){0.1}
\pscircle*(1.25,2.25){0.1}
\pscircle*(1.25,0.25){0.1}

\psline(0,0)(0,2)(2,2)(2,0)(0,0)
\psline(0,0)(-0.75,0.25)(-0.75,2.25)(1.25,2.25)(2,2)
\psline(0,2)(-0.75,2.25)

\psset{linestyle=dashed}

\psline(1.25,0.25)(1.25,2.25)
\psline(1.25,0.25)(-0.75,0.25)
\psline(1.25,0.25)(2,0)

\rput(0.2,-0.3){$x$}
\rput(2.2,-0.3){$x_1$}
\rput(0.3,1.7){$x_3$}
\rput(2.4,1.7){$x_{13}$}

\rput(-1.1,0.5){$x_2$}
\rput(0.9,0.5){$x_{12}$}
\rput(0.9,2.5){$x_{123}$}
\rput(-1.1,2.5){$x_{23}$}

\end{pspicture}

%% file: 3-DRC/DRC.tex
In this section, we review fundamental projective geometry configuration of double reflection in the $d$-dimensional projective space $\mathbf{P}^d$ over an arbitrary field of characteristic not equal to $2$.
Detailed discussion on this can be found in \cite{DragRadn2008,DragRadn2011book} (see also \cite{CCS1993}).

The section is concluded by Proposition \ref{prop:drc.quad}, where we show that the double configuration can take the role of the quad-equation, that is every line in such a configuration is determined by the remaining three.

Let us start with recalling the notions of quadrics and confocal families in the projective space.

\emph{A quadric} in $\mathbf{P}^d$ is the set given by equation of the form:
$$
(Q\xi,\xi)=0,
$$
where $Q$ is a symmetric $(d+1)\times(d+1)$ matrix, and $\xi=[\xi_0:\xi_1:\dots:\xi_d]$ are homogeneous coordinates of a point in the space.

Assume two quadrics are given:
$$
\mathcal{Q}_1\ :\ (Q_1\xi,\xi)=0,
\qquad
\mathcal{Q}_2\ :\ (Q_2\xi,\xi)=0.
$$
\emph{A pencil of quadrics} is the family of quadrics given by equations:
$$
\left((Q_1+\lambda Q_2)\xi,\xi\right)=0,
\quad
\lambda\in\mathbf{P}^1.
$$

\emph{A confocal system of quadrics} is a family of quadrics such that its projective dual is a pencil of quadrics.

Now, let us recall the definition of reflection on a quadric in the projective space, where metrics is not defined.
This definition, together with its crucial properties -- One Reflection Theorem and Double Reflection Theorem, can be found in \cite{CCS1993}.


Denote by $u$ the tangent plane to $\mathcal Q_1$ at point $x$ and by $z$ the pole
of $u$ with respect to $\mathcal Q_2$.
Suppose lines $\ell_1$ and $\ell_2$ intersect at $x$, and the plane containing these two lines meet $u$ along $\ell$.

\begin{definition}
If lines $\ell_1$, $\ell_2$, $xz$, $\ell$ are coplanar and harmonically conjugated, we say that rays $\ell_1$ and $\ell_2$ \emph{obey the reflection law} at the point $x$ of the quadric $\mathcal{Q}_1$ with respect to the confocal system which contains 
$\mathcal{Q}_1$ and $\mathcal{Q}_2$.
\end{definition}

If we introduce a coordinate system in which quadrics $\mathcal{Q}_1$ and
$\mathcal{Q}_2$ are confocal in the usual sense, reflection defined in this way is same as the standard, metric one.

\begin{theorem}[One Reflection Theorem]\label{th:ORT}
Suppose lines $\ell_1$ and $\ell_2$ obey the reflection law at point $x$ of
$\mathcal{Q}_1$ with respect to the confocal system determined by quadrics
$\mathcal{Q}_1$ and $\mathcal{Q}_2$.
Let $\ell_1$ intersect $\mathcal{Q}_2$ at $y_1'$ and $y_1$, $u$ be a tangent plane to 
$\mathcal{Q}_1$ at $x$, and $z$ its pole with respect to $\mathcal{Q}_2$.
Then lines $y_1'z$ and $y_1z$ respectively contain intersecting points $y_2'$ and $y_2$ of line $\ell_2$ with $\mathcal{Q}_2$.
Converse is also true.
\end{theorem}

\begin{corollary}\label{cor:ORT}
Let lines $\ell_1$ and $\ell_2$ obey the reflection law of $\mathcal{Q}_1$ with respect to the confocal system determined by quadrics $\mathcal{Q}_1$ and $\mathcal{Q}_2$.
Then $\ell_1$ is tangent to $\mathcal{Q}_2$ if and only if is tangent $\ell_2$ to 
$\mathcal{Q}_2$;
$\ell_1$ intersects $\mathcal{Q}_2$ at two points if and only if $\ell_2$ intersects
$\mathcal{Q}_2$ at two points.
\end{corollary}

Next theorem is crucial for our further considerations.

\begin{theorem}[Double Reflection Theorem]\label{th:DRT}
Suppose that $\mathcal{Q}_1$, $\mathcal{Q}_2$ are given quadrics and $\ell_1$ line intersecting $\mathcal{Q}_1$ at the point $x_1$ and $\mathcal{Q}_2$ at $y_1$.
Let $u_1$ be the tangent plane to $\mathcal{Q}_1$ at $x_1$;
$z_1$ the pole of $u_1$ with respect to $\mathcal{Q}_2$;
$v_1$ the tangent plane $\mathcal{Q}_2$ at $y_1$;
and by $w_1$ the pole of $v_1$ with respect to $\mathcal{Q}_1$.
Denote by $x_2$ the intersecting point of line $w_1x_1$ with $\mathcal{Q}_1$,
$x_2\neq x_1$;
by $y_2$ intersection of $y_1z_1$ with $\mathcal{Q}_2$, $y_2\neq y_1$;
and $\ell_2=x_1y_2$, $\ell_1'=y_1x_2$, $\ell_2'=x_2y_2$.

Then pair $\ell_1$, $\ell_2$ obey the reflection law at $x_1$ of $\mathcal{Q}_1$;
$\ell_1$, $\ell_1'$ obey the reflection law at $y_1$ of $\mathcal{Q}_2$;
$\ell_2$, $\ell_2'$ obey the reflection law at $y_2$ of $\mathcal{Q}_2$;
and $\ell_1'$, $\ell_2'$ obey the reflection law at point  $x_2$ of $\mathcal{Q}_1$.
\end{theorem}

Let us remark that the four tangent planes at the reflection points belong to a pencil.

\begin{corollary}\label{cor:refl}
If line $\ell_1$ is tangent to quadric $\mathcal{Q}'$ confocal with $\mathcal{Q}_1$ and $\mathcal{Q}_2$, then rays $\ell_2$, $\ell_1'$, $\ell_2'$ also touch $\mathcal{Q}'$.
\end{corollary}

The following definition of virtual reflection configuration and double reflection configurations is from \cite{DragRadn2008}.
These configurations played central role in that work.
In Theorem \ref{th:virt.refl}, important properties of these configurations is given.
This theorem is proved in \cite{DragRadn2008}.

The configurations are connected with so-called real and virtual reflections, but their properties remain in the projective case, too.

Let points  $x_1$, $x_2$ belong to $\mathcal{Q}_1$ and  $y_1$, $y_2$ to $\mathcal{Q}_2$.

\begin{definition}\label{def:VRC}
We will say that the quadruple of points $x_1, x_2, y_1, y_2$ constitutes 
\emph{a virtual reflection configuration} if pairs of lines 
$x_1 y_1$, $x_1 y_2$; $x_2 y_1$, $x_2 y_2$; $x_1 y_1$, $x_2 y_1$; $x_1 y_2$, $x_2 y_2$ satisfy the reflection law at points
$x_1$, $x_2$ of $\mathcal Q_1$ and $y_1$, $y_2$ of $\mathcal Q_2$
respectively, with respect to the confocal system determined by
$\mathcal Q_1$ and $\mathcal Q_2$.

If, additionally, the tangent planes to $\mathcal Q_1, \mathcal Q_2$
at $x_1, x_2$; $y_1, y_2$ belong to a pencil, we say that these
points constitute \emph{a double reflection configuration} (see
Figure \ref{fig:virtual_reflection}).
\end{definition}

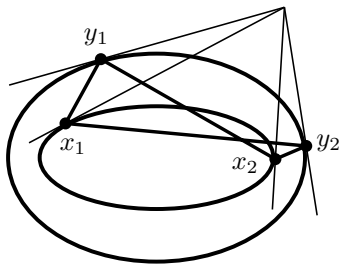
\begin{figure}[h]
\centering
\input{0-figures/virtual}
\caption{Double reflection configuration}\label{fig:virtual_reflection}
\end{figure}

Now, we list some of the basic facts about the double reflection
configurations.

\begin{theorem}\label{th:virt.refl}
Let $\mathcal Q_1$, $\mathcal Q_2$ be two quadrics in the projective space
$\mathbf{P}^d$, $x_1$, $x_2$ points on $\mathcal Q_1$ and $y_1$, $y_2$ on 
$\mathcal{Q}_2$.
If the tangent hyperplanes at these points to the quadrics belong to a pencil, then $x_1, x_2, y_1, y_2$
constitute a virtual reflection configuration.

Furthermore, suppose that the projective space is defined over the field of reals.
Introduce a coordinate system, such that $\mathcal{Q}_1$, $\mathcal{Q}_2$ become confocal ellipsoids in the Euclidean space.
If $\mathcal{Q}_1$ is placed inside $\mathcal{Q}_2$, then the sides of the quadrilateral $x_1y_1x_2y_2$ obey the real reflection from $\mathcal{Q}_2$ and the virtual reflection from $\mathcal{Q}_1$.
\end{theorem}

The statement converse to Theorem \ref{th:virt.refl} is the following
\begin{proposition}\label{prop:drc-ellipsoids}
In the Euclidean space $\mathbf{E}^d$, two confocal ellipsoids $\mathcal{E}_1$ and $\mathcal{E}_2$ are given.
Let points $X_1$, $X_2$ belong to $\mathcal{E}_1$, $Y_1$, $Y_2$ to $\mathcal{E}_2$, and let $\alpha_1$, $\alpha_2$, $\beta_1$, $\beta_2$ be the corresponding tangent planes.
If a quadruple $X_1, X_2, Y_1, Y_2$ is a virtual reflection configuration, then planes $\alpha_1$, $\alpha_2$, $\beta_1$, $\beta_2$ belong to a pencil.
\end{proposition}

The next proposition shows that three lines of a double reflection configuration uniquely determine the fourth one.

\begin{proposition}\label{prop:drc.quad}
Let $\ell$, $\ell_1$, $\ell_2$ be lines and $\mathcal{Q}_1$, $\mathcal{Q}_2$ quadrics in the projective space.
Suppose that $\ell$, $\ell_1$ satisfy the reflection law on $\mathcal{Q}_1$,
and $\ell$, $\ell_2$ the reflection law on $\mathcal{Q}_2$, with respect to the confocal system determined by these two quadrics.
Then there is unique line $\ell_{12}$ such that four lines 
$\ell$, $\ell_1$, $\ell_2$, $\ell_{12}$ form a double reflection configuration.
\end{proposition}

\begin{proof}
Let $\alpha$ be the tangent hyperplane to $\mathcal{Q}_1$ at the intersection point of 
$\ell$ and $\ell_1$, and $\beta$ the tangent to $\mathcal{Q}_2$ at the intersection of 
$\ell$ and $\ell_2$.
In the pencil of hyperplanes determined by $\alpha$ and $\beta$, take tangents 
$\alpha_1$ to $\mathcal{Q}_1$ and $\beta_1$ to $\mathcal{Q}_2$,
$\alpha_1\neq\alpha$, $\beta_1\neq\beta$.
By Theorem \ref{th:virt.refl}, $\ell_{12}$ will be the line containing touching points of $\alpha_1$ and $\beta_1$ with the corresponding quadrics.
\end{proof}

\begin{remark}\label{rem:drc.quad}
Proposition \ref{prop:drc.quad} shows that double reflection configuration is playing the role of the quad-equation for lines in the projective space.
\end{remark}

In next section, we are going to show that double reflection configuration is $3D$-consistent.

%% file: 0-figures/virtual.tex
\begin{pspicture}(-2.5,-2)(2.5,2.5)
\psset{fillstyle=none, linecolor=black, linestyle=solid, linewidth=1.5pt}

\psellipse(0,0)(2,1.41)
\psellipse(0,0)(1.58,0.71)

\psset{linewidth=0.6pt}

\psline(1.7,2)(1.99,0.155)(2.135,-0.7675)
\psline(1.7,2)(-0.74,1.31)(-1.96,0.965)
\psline(1.7,2)(1.58,-0.02)(1.54,-0.69)
\psline(1.7,2)(-1.21,0.455)(-1.695,0.197)

\psset{linewidth=1.3pt}

\psline(1.99,0.155)(1.58,-0.02)(-0.74,1.31)(-1.21,0.455)(1.99,0.155)

\pscircle*(1.99,0.155){0.08}
\pscircle*(1.58,-0.02){0.08}
\pscircle*(-0.74,1.31){0.08}
\pscircle*(-1.21,0.455){0.08}

\rput(-1.11,0.155){$x_1$}
\rput(-0.8,1.61){$y_1$}
\rput(2.29,0.2){$y_2$}
\rput(1.18,-0.1){$x_2$}

\end{pspicture}

%% file: 4-algebra/billiard-algebra1.tex
In this section we make a quad-graph interpretation of some results obtained using billiard algebra.
This algebra is constructed by the authors -- for details of the construction, see 
\cite{DragRadn2008,DragRadn2011book}.

Let us start from a theorem on confocal families of quadrics from \cite{DragRadn2008}:

\begin{theorem}[Six-pointed star theorem]\label{th:zvezda}
Let $\mathcal F$ be a family of confocal quadrics in $\mathbf P^3$.
There exist configurations consisting of twelve planes in $\mathbf{P}^3$ with the following properties:
\begin{itemize}
\item
The planes may be organized in eight triplets, such that each plane in a triplet is tangent to a different quadric from $\mathcal F$ and the three touching points are collinear.
Every plane in the configuration is a member of two triplets.

\item
The planes may be organized in six quadruplets, such that the planes in each quadruplet belong to a pencil and they are tangent to two different quadrics from $\mathcal F$.
Every plane in the configuration is a member of two quadruplets.
\end{itemize}
Moreover, such a configuration is determined by three planes tangent to three different quadrics from $\mathcal F$, with collinear touching points.
\end{theorem}

The configuration of the planes in the dual space $\mathbf{P}^{3*}$ is shown on Figure \ref{fig:zvezda}: each plane corresponds to a vertex of the polygonal line.
\begin{figure}[h]
\centering
\input{0-figures/zvezda}
\caption{The configuration of planes}\label{fig:zvezda}
\end{figure}
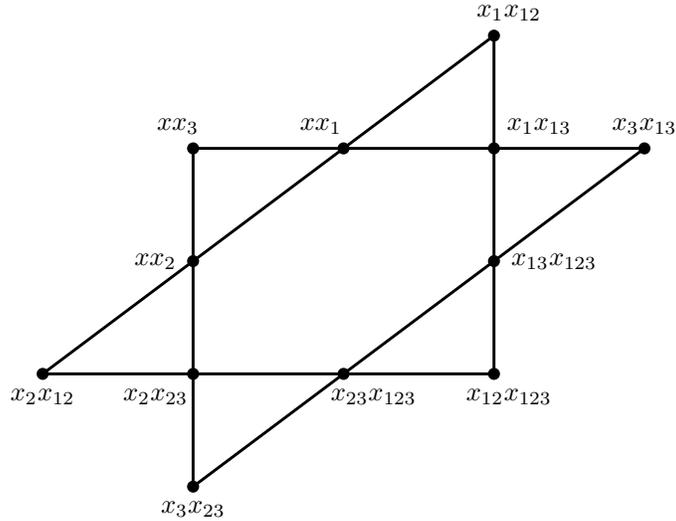

To understand the notation used on Figure \ref{fig:zvezda}, let us recall the construction leading to configurations from Theorem \ref{th:zvezda}.
Take $\mathcal{Q}_1$, $\mathcal{Q}_2$, $\mathcal{Q}_3$ to be quadrics from 
$\mathcal{F}$, and $\alpha$, $\beta$, $\gamma$ respectively their tangent planes such that the touching points $A$, $B$, $C$ are collinear.
Denote by $x$ the line containing these three points, and by $x_1$, $x_2$, $x_3$ lines obtained from $x$ by reflections of $\mathcal{Q}_1$, $\mathcal{Q}_2$, $\mathcal{Q}_3$ at $A$, $B$, $C$ respectively.

Now, as in Proposition \ref{prop:drc.quad}, determine lines
$x_{12}$, $x_{13}$, $x_{23}$, $x_{123}$ such that they respectively complete triplets
$\{x,x_1,x_2\}$, $\{x,x_1,x_3\}$, $\{x,x_2,x_3\}$, $\{x_3,x_{13},x_{23}\}$
to double reflection configurations.\footnote{Let us note that in \cite{DragRadn2008}, lines $x$, $x_1$, $x_2$, $x_3$, $x_{12}$, $x_{13}$, $x_{23}$, $x_{123}$ were respectively denoted by $\mathcal{O}$, $p$, $q$, $s$, $-x$, $p_1$, $q_1$, $x+s$, where the addition is defined in the billiard algebra introduced in that paper.}

Notice the following objects on Figure \ref{fig:zvezda}:
\begin{description}
\item[twelve verteces]
to each vertex a plane tangent to one of three quadrics
$\mathcal{Q}_1$, $\mathcal{Q}_2$, $\mathcal{Q}_3$
and a pair of lines are joined --- the lines of any pair are reflected to each other on the corresponding tangent plane;

\item[eight triangles]
three planes joined to verteces of any triangle are tangent to three quadrics
$\mathcal{Q}_1$, $\mathcal{Q}_2$, $\mathcal{Q}_3$ with touching points being collinear --- thus to each triangle, the line containing these points is naturaly joined ;

\item[six edges]
each edge containes four verteces --- four planes joined to these verteces are in the same pencil and a double reflection configuration corresponds to each edge.
\end{description}

Now, we are ready to prove the $3D$-consistency of the quad-relation introduced via double reflection configurations.

\begin{theorem}\label{th:3d}
Let $x$, $x_1$, $x_2$, $x_3$ be lines in the projective space, such that $x_1$, $x_2$, $x_3$ are obtained from $x$ by reflections on confocal quadrics
$\mathcal{Q}_1$, $\mathcal{Q}_2$, $\mathcal{Q}_3$ respectively.
Introduce lines $x_{12}$, $x_{13}$, $x_{23}$, $x_{123}$ such that the following quadruplets are double reflection configurations:
\begin{gather*}
\{x,x_1,x_{12},x_2\}, \quad
\{x,x_1,x_{13},x_3\}, \quad
\{x,x_2,x_{23},x_3\}, \quad
\{x_1,x_{12},x_{123},x_{13}\}.
\end{gather*} 
Then the following quadruplets are also double reflection configurations:
$$
\{x_2,x_{12},x_{123},x_{23}\},\quad
\{x_3,x_{13},x_{123},x_{23}\}.
$$ 
\end{theorem}
\begin{proof}
Let us remark that the configuration described in Theorem \ref{th:zvezda} has obviously a combinatorial structure of cube, with planes corresponding to edges of the cube.
In this way, lines $x$, $x_1$, $x_2$, $x_3$, $x_{12}$, $x_{13}$, $x_{23}$, $x_{123}$ will correspond to vertices of the cube as shown on Figure \ref{fig:cube}.
A pair of lines is represented by endpoints of an edge if they reflect to each other on the plane joined to this edge.
Faces of the cube represent double reflection configurations.
Notice also that planes joined to parallel edges of the cube are tangent to the same quadric.
The statement follows from Theorem \ref{th:zvezda} and the construction given after, see
Figure \ref{fig:zvezda}.
\end{proof}

%% file: 0-figures/zvezda.tex
\begin{pspicture}(-4.5,-3.5)(4.5,3.5)
\psset{fillstyle=none, linecolor=black, linestyle=solid, linewidth=1.1pt}

\psline(-4,-1.5)(2,-1.5)(2,3)(-4,-1.5)
\psline(-2,1.5)(-2,-3)(4,1.5)(-2,1.5)

\pscircle*(-4,-1.5){0.08}
\rput(-4,-1.8){$x_2x_{12}$}

\pscircle*(-2,-1.5){0.08}
\rput(-2.5,-1.8){$x_2x_{23}$}

\pscircle*(0,-1.5){0.08}
\rput(0.4,-1.8){$x_{23}x_{123}$}

\pscircle*(2,-1.5){0.08}
\rput(2.2,-1.8){$x_{12}x_{123}$}

\pscircle*(4,1.5){0.08}
\rput(4,1.8){$x_3x_{13}$}

\pscircle*(2,1.5){0.08}
\rput(2.6,1.8){$x_1x_{13}$}

\pscircle*(0,1.5){0.08}
\rput(-0.3,1.8){$xx_1$}

\pscircle*(-2,1.5){0.08}
\rput(-2.2,1.8){$xx_3$}

\pscircle*(-2,0){0.08}
\rput(-2.5,0){$xx_2$}

\pscircle*(2,0){0.08}
\rput(2.8,0){$x_{13}x_{123}$}

\pscircle*(2,3){0.08}
\rput(2.2,3.3){$x_1x_{12}$}

\pscircle*(-2,-3){0.08}
\rput(-2,-3.3){$x_{3}x_{23}$}

\end{pspicture}

%% file: 5-DRnets/drnets.tex
Assume a pencil of confocal quadrics is given in $\mathbf{P}^d$, and fix $d-1$ quadrics from the pencil.
Take $\mathcal{A}\subset\mathcal{L}^d$ to be the set of all lines touching these 
$d-1$ quadrics.

Notice that, by the Chasles theorem \cite{Chasles}, every line $\mathbf{P}^d$ is touching $d-1$ quadrics from a given confocal family.
Moreover, by Corollary \ref{cor:refl}, these $d-1$ quadrics are preserved by billiard reflection.
Confocal quadrics touched by a line are called \emph{caustics} of this line, or consequently, caustics of the billiard trajectory that contains the line.

\begin{definition}\label{def:dr-net}
\emph{A double reflection net} is a map 
\begin{equation}\label{eq:drnet}
\varphi\ :\ \mathbf{Z}^m \to \mathcal{A},
\end{equation}
such that there exist $m$ quadrics $\mathcal{Q}_1$, \dots, $\mathcal{Q}_m$ from the confocal pencil, satisfying the following conditions:
\begin{enumerate}
\item
sequence $\{\varphi(\mathbf{n}_0+i\mathbf{e}_j)\}_{i\in\mathbf{Z}}$ represents a billiard trajectory within $\mathcal{Q}_j$, for each $j\in\{1,\dots,m\}$ and 
$\mathbf{n}_0\in\mathbf{Z}^m$;

\item
lines $\varphi(\mathbf{n}_0)$, $\varphi(\mathbf{n}_0+\mathbf{e}_i)$,
$\varphi(\mathbf{n}_0+\mathbf{e}_j)$, $\varphi(\mathbf{n}_0+\mathbf{e}_i+\mathbf{e}_j)$ form a double reflection configuration, for all $i,j\in\{1,\dots,m\}$, $i\neq j$ and $\mathbf{n}_0\in\mathbf{Z}^m$.
\end{enumerate}
\end{definition}

In other words, for each edge in $\mathbf{Z}^m$ of direction $e_i$,
the lines corresponding to its vertices intersect at $\mathcal{Q}_i$,
and four tangent planes at the intersection points, associated to an elementary quadrilateral, belong to a pencil.

In the following Subsections \ref{sec:mink}--\ref{sec:const}, we describe some examples of double reflection nets.
After that, in Subsection \ref{sec:F}, we define $F$-transformations of double reflection nets and we conclude this section by establishing connection with Grassmanian Darboux nets from \cite{ABS2009} in Subsection \ref{sec:grass}.

\subsection{Example of a double reflection net in the Minkowski space}\label{ex:hip-tropic-light}
\label{sec:mink}
\input{5-DRnets/example-minkowski}

\subsection{Poncelet-Darboux grids and double reflection nets}
\input{5-DRnets/example-darboux}

\subsection{$s$-skew lines and double reflection nets}
\label{sec:skew}
\input{5-DRnets/example-skew}

\subsection{Construction of double reflection nets}
\label{sec:const}
\input{5-DRnets/example-const}

\subsection{Focal nets and F-transformations of double reflection nets}
\label{sec:F}
\input{5-DRnets/F}

\subsection{Double reflection nets and Grassmannian Darboux nets}
\label{sec:grass}
\input{5-DRnets/grassmannian}

%% file: 5-DRnets/example-minkowski.tex
Consider three-dimensional Minkowski space $\mathbf{E}^{2,1}$, that is $\mathbf{R}^3$ with the Minkowski scalar product:
$$
\langle (x_1,y_1,z_1),(x_2,y_2,z_2)\rangle = x_1x_2 + y_1y_2 -z_1z_2.
$$
In this space, let a general confocal family be given (see \cite{KhTab2009}):
\begin{equation}\label{eq:confocal.quadrics3}
\mathcal{Q}_{\lambda}\ :\ \frac{x^2}{a-\lambda}+\frac{y^2}{b-\lambda}+\frac{z^2}{c+\lambda}=1,
\quad
\lambda\in\mathbf{R}.
\quad
(a>b>0,\ c>0),  
\end{equation}

In this family, one can observe four different geometric types of non-degenerated quadrics:
\begin{itemize}
\item
one-sheeted hyperboloids oriented along $z$-axis, for $\lambda\in(-\infty,-c)$;
\item
ellipsoids, corresponding to $\lambda\in(-c,b)$;
\item
one-sheeted hyperboloids oriented along $y$-axis, for $\lambda\in(b,a)$ (see Figure \ref{fig:hip-tropic-light});
\item
two-sheeted hyperboloids, for $\lambda\in(a,+\infty)$ -- these hyperboloids are oriented along $z$-axis.
\end{itemize}
In addition, there are four degenerated quadrics: $\mathcal{Q}_{a}$, $\mathcal{Q}_{b}$, $\mathcal{Q}_{-c}$, $\mathcal{Q}_{\infty}$, that is planes $x=0$, $y=0$, $z=0$, and the plane at the infinity respectively.

\begin{figure}[h]
\centering
\includegraphics[width=6cm, height=8.4cm]{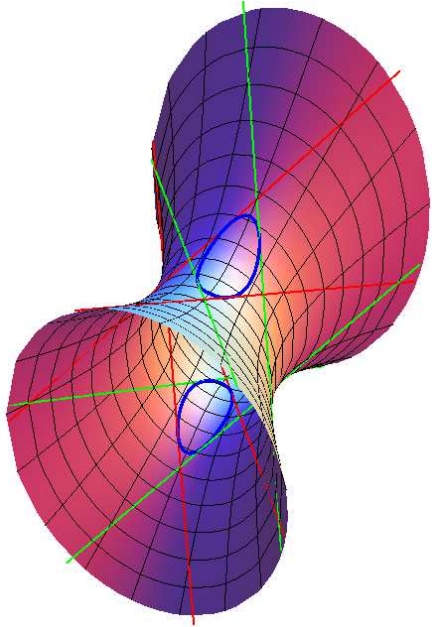}
\caption{The tropic curves and their light-like tangents on a one-sheeted hyperboloid oriented along $y$-axis from confocal family (\ref{eq:confocal.quadrics3})}\label{fig:hip-tropic-light}
\end{figure}

On each quadric, notice the \emph{tropic curves} -- the set of points where the induced metrics on the tangent plane is degenerate.

It is known that (see \cite{DragRadn2011}):
\begin{itemize}
\item
a tangent line to the tropic curve of a non-degenarate quadric of the family (\ref{eq:confocal.quadrics3}) is always space-like, except on a one-sheeted hyperboloid oriented along $y$-axis;
\item
tangent lines of a tropic on one-sheeted hyperboloids oriented along $y$-axis are light-like exactely at four points, while at other points of the tropic curve, the tangents are space-like;
\item
a tangent line to the tropic of a quadric from (\ref{eq:confocal.quadrics3}) belongs to the quadric if and only if it is light-like;
\item
each one-sheeted hyperboloid oriented along $y$-axis has exactly eight light-like generatrices, as shown on Figure \ref{fig:hip-tropic-light}.
\end{itemize}

Fix $\lambda_0\in(b,a)$, and consider hyperboloid $\mathcal{Q}_{\lambda_0}$.
Denote by $a_1$, $a_2$, $a_3$, $a_4$, $b_1$, $b_2$, $b_3$, $b_4$ the light-like generatrices of $\mathcal{Q}_{\lambda_0}$, in the following way (see Figure \ref{fig:tropic.tangents}):
\begin{figure}[h]
\centering
\input{0-figures/light-tangents}
\caption{The tropic curves of $\mathcal{Q}_{\lambda_0}$ and their light-like tangents.}\label{fig:tropic.tangents}
\end{figure}
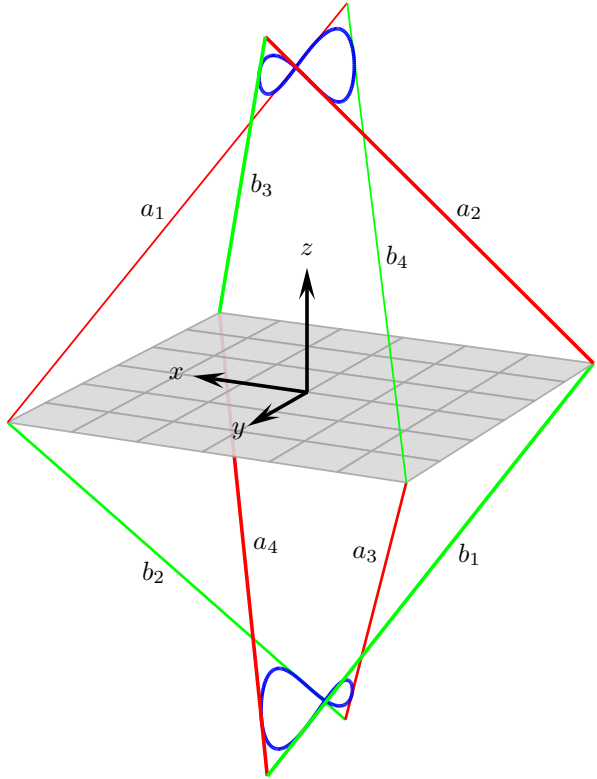
\begin{itemize}
\item
lines $a_i$ belong to one, and $b_i$ to the other family of generatrices of $\mathcal{Q}_{\lambda_0}$; that is, $a_i$ and $a_j$ are always skew for $i\neq j$, while $a_i$ and $b_j$ are coplanar for all $i,j$;

\item
$a_1$, $a_2$, $b_3$, $b_4$ are tangent to the tropic curve contained in the half-space $z>0$, while $a_3$, $a_4$, $b_1$, $b_2$ are touching the other tropic curve;

\item
$a_i$ is parallel to $b_i$ for each $i$;

\item
pairs $(a_1,b_2)$, $(a_2,b_1)$, $(a_3,b_4)$, $(a_4,b_3)$ have intersection points in the $xy$-plane;

\item
pairs $(a_1,b_3)$, $(a_2,b_4)$, $(a_3,b_1)$, $(a_4,b_2)$ have intersection points in the $xz$-plane;

\item
pairs $(a_1,b_4)$, $(a_2,b_3)$, $(a_3,b_2)$, $(a_4,b_1)$ have intersection points in the $yz$-plane.
\end{itemize}

Take $\mathcal{A}$ to be the set of all generatrices of hyperboloid 
$\mathcal{Q}_{\lambda_0}$, i.e.~ the set of all lines having $\mathcal{Q}_{\lambda_0}$ as a double caustic.
In particular, $\mathcal{A}$ contains all lines $a_i$, $b_i$.

It is possible to define map
$$
\varphi_M\ :\ \mathbf{Z}^4\to\mathcal{A},
$$
such that the image of $\varphi_M$ is the set $\{a_1,a_2,a_3,a_4,b_1,b_2,b_3,b_4\}$ and
for each $\mathbf{n}\in\mathbf{Z}^4$ lines $\varphi_{M}(\mathbf{n}+\mathbf{e}_1)$,
$\varphi_{M}(\mathbf{n}+\mathbf{e}_2)$, $\varphi_{M}(\mathbf{n}+\mathbf{e}_3)$,
$\varphi_{M}(\mathbf{n}+\mathbf{e}_4)$ are obtained from $\varphi_{M}(\mathbf{n})$ by reflection on $\mathcal{Q}_{a}$, $\mathcal{Q}_{b}$, $\mathcal{Q}_{-c}$, $\mathcal{Q}_{\infty}$ respectively.

More precisely, $\varphi_M$ will be periodic with period $2$ in each coordinate and:
\begin{gather*}
\varphi_M(0,0,0,0)=\varphi_M(1,1,1,1)=a_1,\quad \varphi_M(1,1,0,0)=\varphi_M(0,0,1,1)=a_2,
\\
\varphi_M(1,0,1,0)=\varphi_M(0,1,0,1)=a_3,\quad \varphi_M(0,1,1,0)=\varphi_M(1,0,0,1)=a_4,
\\
\varphi_M(0,0,0,1)=\varphi_M(1,1,1,0)=b_1,\quad \varphi_M(1,1,0,1)=\varphi_M(0,0,1,0)=b_2,
\\
\varphi_M(1,0,1,1)=\varphi_M(0,1,0,0)=a_3,\quad \varphi_M(0,1,1,1)=\varphi_M(1,0,0,0)=b_4,
\end{gather*}

It is shown on Figure \ref{fig:hiper-kocka} how vertices of the unit tesseract in $\mathbf{Z}^4$ are mapped by $\varphi_M$.

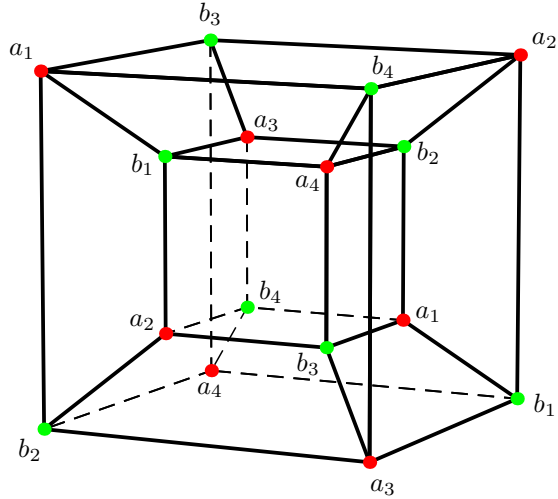
\begin{figure}[h]
\centering
\input{0-figures/hiper-kocka}
\caption{Mapping $\varphi_M$ on the unit tesseract. }\label{fig:hiper-kocka}
\end{figure}

It is straightforward to prove the following
\begin{proposition}
$\varphi_M$ is a double reflection net.
\end{proposition}

%% file: 0-figures/light-tangents.tex
\psset{unit=1.2}
\begin{pspicture*}(-4.5,-4.5)(4.5,4.5)

\psset{viewpoint=-15 30 10,Decran=50,lightsrc=-5 10 50}


\psSolid[object=line,linecolor=green,linewidth=0.03,
args=0 -0.693 -2.946 1.673 1.732 0]

\psSolid[object=line,linecolor=red,linewidth=0.03,
args=0 -0.693 -2.946 -1.673 1.732 0]

\defFunction[algebraic]{tropic2}(t)
{(5-4.2)/(sqrt(5+2))*cos(t)}
{(3-4.2)/(sqrt(3+2))*sin(t)}
{-(2+4.2)*sqrt( ( (3+2)*(cos(t))^2+(5+2)*(sin(t))^2 )/((5+2)*(3+2)) ) }

\psSolid[object=courbe,
linecolor=blue,linewidth=0.04,
r=0,range=0 pi 2 mul, resolution=360,
function=tropic2](0,0,0)

\psSolid[object=line,linecolor=red,linewidth=0.03,
args=0 0.693 -2.946 1.673 -1.732 0]

\psSolid[object=line,linecolor=green,linewidth=0.03,
args=0 0.693 -2.946 -1.673 -1.732 0]

\psSolid[object=plan,
definition=equation, linewidth=0.02, linecolor=gray!70,
fillcolor=gray!30,opacity=0.9,ngrid=6 6,
args={[0 0 1 0]},
base=-1.673 1.673 -1.732 1.732]

\psSolid[object=line,linecolor=red,linewidth=0.03,
args=0 -0.693 2.946 1.673 1.732 0]

\psSolid[object=line,linecolor=green,linewidth=0.03,
args=0 -0.693 2.946 -1.673 1.732 0]

\defFunction[algebraic]{tropic1}(t)
{(5-4.2)*cos(t)/(sqrt(5+2))}
{(3-4.2)/(sqrt(3+2))*sin(t)}
{(2+4.2)*sqrt( ( (3+2)*(cos(t))^2+(5+2)*(sin(t))^2 )/((5+2)*(3+2)) ) }

\psSolid[object=courbe,
linecolor=blue,linewidth=0.04,
r=0,range=0 pi 2 mul,
resolution=360, function=tropic1]

\psSolid[object=line,linecolor=red,linewidth=0.03,
args=0 0.693 2.946 -1.673 -1.732 0]

\psSolid[object=line,linecolor=green,linewidth=0.03,
args=0 0.693 2.946 1.673 -1.732 0]

\axesIIID[linewidth=0.04,arrowsize=2pt 3,arrowlength=2,arrowinset=0.4](0,0,0)(1,1,1)

\rput(-1.7,2){$a_1$}
\rput(1.8,2){$a_2$}

\rput(-0.5,2.3){$b_3$}
\rput(1,1.5){$b_4$}

\rput(-1.7,-2){$b_2$}
\rput(1.8,-1.8){$b_1$}

\rput(-0.45,-1.7){$a_4$}
\rput(0.65,-1.8){$a_3$}

\end{pspicture*}

%% file: 0-figures/hiper-kocka.tex
\psset{unit=1}
\begin{pspicture*}(-4,-3.5)(4,3.5)

\psset{viewpoint=-15 30 5,Decran=50,lightsrc=-5 10 50}

\psset{linewidth=0.05}

\psSolid[
object=cube,
a=1.6,
action=draw]

\psSolid[
object=cube,
a=3.2,
action=draw]

\psSolid[object=line,
args=0.8 0.8 0.8 1.6 1.6 1.6]
\psSolid[object=line,
args=-0.8 0.8 0.8 -1.6 1.6 1.6]
\psSolid[object=line,
args=0.8 -0.8 0.8 1.6 -1.6 1.6]
\psSolid[object=line,
args=0.8 0.8 -0.8 1.6 1.6 -1.6]
\psSolid[object=line,
args=-0.8 -0.8 0.8 -1.6 -1.6 1.6]
\psSolid[object=line,
args=-0.8 0.8 -0.8 -1.6 1.6 -1.6]
\psSolid[object=line,linewidth=0.02,linestyle=dashed,
args=0.8 -0.8 -0.8 1.6 -1.6 -1.6]
\psSolid[object=line,
args=-0.8 -0.8 -0.8 -1.6 -1.6 -1.6]

\psset{linecolor=green}

\psSolid[object=point,
args=0.8 0.8 0.8]
\psSolid[object=point,
args=0.8 -0.8 -0.8]
\psSolid[object=point,
args=-0.8 -0.8 0.8]
\psSolid[object=point,
args=-0.8 0.8 -0.8]

\psSolid[object=point,
args=-1.6 1.6 1.6]
\psSolid[object=point,
args=1.6 -1.6 1.6]
\psSolid[object=point,
args=-1.6 -1.6 -1.6]
\psSolid[object=point,
args=1.6 1.6 -1.6]

\psset{linecolor=red}

\psSolid[object=point,
args=-0.8 0.8 0.8]
\psSolid[object=point,
args=-0.8 -0.8 -0.8]
\psSolid[object=point,
args=0.8 -0.8 0.8]
\psSolid[object=point,
args=0.8 0.8 -0.8]

\psSolid[object=point,
args=1.6 1.6 1.6]
\psSolid[object=point,
args=-1.6 -1.6 1.6]
\psSolid[object=point,
args=1.6 -1.6 -1.6]
\psSolid[object=point,
args=-1.6 1.6 -1.6]

\rput(-1.9,-1.1){$a_2$}
\rput(0.3,-1.65){$b_3$}
\rput(-1.9,1){$b_1$}
\rput(0.3,0.75){$a_4$}
\rput(1.9,-1){$a_1$}
\rput(-0.2,-0.7){$b_4$}
\rput(-0.25,1.55){$a_3$}
\rput(1.9,1.15){$b_2$}

\rput(-3.4,-2.8){$b_2$}
\rput(-3.5,2.5){$a_1$}
\rput(-1,3){$b_3$}
\rput(3.45,2.65){$a_2$}
\rput(3.45,-2.2){$b_1$}
\rput(-1,-2){$a_4$}
\rput(1.3,-3.3){$a_3$}
\rput(1.3,2.3){$b_4$}
\end{pspicture*}

%% file: 5-DRnets/example-darboux.tex
Let $\mathcal{E}$ be an ellipse in the Euclidean plane:
$$
\mathcal{E}\ :\ \frac{x^2}{a}+\frac{y^2}{b},\quad a>b>0,
$$
and $(a_i)_{i\in\mathbf{Z}}$ a billiard trajectory within $\mathcal{E}$.

As it is well known, all lines $a_i$ are touching the same conic $\mathcal{C}$ confocal with $\mathcal{E}$.
Here, we will additionally suppose that $\mathcal{C}$ is an ellipse.
Denote by $\mathcal{A}$ the set of tangents of $\mathcal{C}$.

Fix $m$ positive integers $k_1$, \dots, $k_{m}$ and define mapping:
$$
\varphi_D\ :\ \mathbf{Z}^{m}\to\mathcal{A},
\quad
\varphi_D(n_1,\dots,n_m)=a_{n_1k_1+\dots+n_mk_{m}}.
$$

\begin{proposition}\label{prop:phiD}
Map $\varphi_D$ is a double reflection net.
\end{proposition}

\begin{proof}
Since $\varphi_D(\mathbf{n}+i\mathbf{e}_j)=a_{n_1k_1+\dots+n_mk_{m}+ik_j}$,
($\mathbf{n}=(n_1,\dots,n_m)$), it follows by \cite[Theorem 18]{DragRadn2011} that sequence $(\varphi_D(\mathbf{n}+i\mathbf{e}_j))_{i\in\mathbf{Z}}$ represents a billiard trajectory within some ellipse $\mathcal{E}_j$, confocal with $\mathcal{E}$ and $\mathcal{C}$.

Immediately, by Definition \ref{def:VRC}, lines $\varphi(\mathbf{n}_0)$, $\varphi(\mathbf{n}_0+\mathbf{e}_i)$, $\varphi(\mathbf{n}_0+\mathbf{e}_j)$, $\varphi(\mathbf{n}_0+\mathbf{e}_i+\mathbf{e}_j)$ form a virtual reflection configuration for each $\mathbf{n}_0\in\mathbf{Z}^m$, $i,j\in\{1,\dots,m\}$.

Moreover, by Propositon \ref{prop:drc-ellipsoids}, they also form a double reflection configuration.
\end{proof}

\begin{remark}
It interesting to consider only nets where $m$ ellipses $\mathcal{E}_j$ appearing in the proof of Proposition \ref{prop:phiD} are distinct.
If some of them coincide, then we may consider a corresponding subnet.

Suppose that $(a_i)$ is a non-periodic trajectory.
Then, choosing any $m$, and any set of distinct positive numbers $k_1$, \dots, $k_m$, we get substantially different double reflection nets.

For $(a_i)$ being $n$-periodic, it is enough to consider the case $k_i=i$, $i\in\{1,\dots,[n/2]\}$, $(m=[n/2])$.
\end{remark}

\begin{example}
Suppose $(a_i)$ is a $5$-perodic billiard trajectory within $\mathcal{E}$, see Figure \ref{fig:petougao}.
\begin{figure}[h]
\centering
\input{0-figures/petougao}
\caption{A Poncelet pentagon.}\label{fig:petougao}
\end{figure}
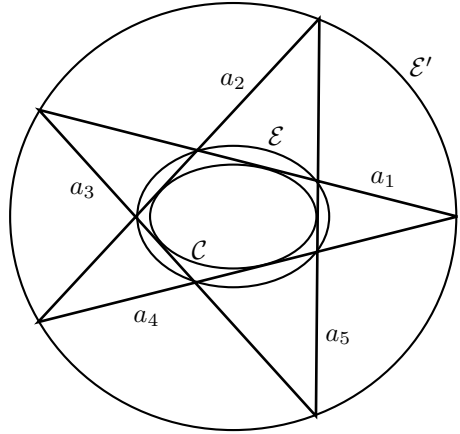

The corresponding double reflection net is:
$$
\varphi_D\ :\ \mathbf{Z}^{2}\to\mathcal{A},
\quad
\varphi_D(n_1,n_2)=a_{n_1+2n_2}.
$$
\end{example}

%% file: 0-figures/petougao.tex
\psset{unit=0.5}
\begin{pspicture*}(-6.3,-6)(6.3,6)
\psset{linewidth=0.8pt}

\psellipse(0,0)(5.96, 5.71)
\psellipse(0,0)(2.58, 1.91)
\psellipse(0,0)(2.24, 1.41)

\psset{linewidth=1pt}
\psline(5.96,0)(-5.17, 2.84)(2.2, -5.31)(2.29, 5.27)(-5.19, -2.81)(5.96,0)

\rput(5,4){$\mathcal{E}'$}
\rput(1.15,2.15){$\mathcal{E}$}
\rput(-0.9,-0.9){$\mathcal{C}$}

\rput(4,1){$a_1$}
\rput(0,3.6){$a_2$}
\rput(-4,0.7){$a_3$}
\rput(-2.3,-2.7){$a_4$}
\rput(2.8,-3.2){$a_5$}

\end{pspicture*}

%% file: 5-DRnets/example-skew.tex
Now, let us consider a family of confocal quadrics in $\mathbf{E}^d$ ($d\ge3$) and fix its $d-1$ quadrics.
As usually, $\mathcal{A}$ is the set of all lines tangent ot the fixed quadrics.

It is shown in \cite{DragRadn2008} that, from a line in $\mathcal{A}$, we can obtain any other line from that set in at most $d-1$ reflections on quadrics from the confocal family.
We called lines $a$, $b$ from $\mathcal{A}$ \emph{$s$-skew} if $s$ is the smallest number such that they can be obtained $s+1$ such reflections.

Now, suppose lines $a$, $b$ are $s$-skew ($s\ge1$), and let
$\mathcal{Q}_1$, \dots, $\mathcal{Q}_{s+1}$
be the corresponding quadrics from the confocal family.

\begin{theorem}
There is a unique double reflection net
$$
\varphi_s\ :\ \mathbf{Z}^{s+1}\to\mathcal{A}
$$
which satisfies the following:
\begin{itemize}
\item
$\varphi_s(0,\dots,0)=a$

\item
$\varphi_s(1,\dots,1)=b$

\item
$\{\varphi(\mathbf{n}_0+i\mathbf{e}_j)\}_{i\in\mathbf{Z}}$ represents a billiard trajectory within $\mathcal{Q}_j$, for each $j\in\{1,\dots,s+1\}$ and $\mathbf{n}_0\in\mathbf{Z}^{s+1}$.
\end{itemize}
\end{theorem}

\begin{proof}
First, using construction of billiard algebra from \cite{DragRadn2008}, we are going to define mapping $\varphi_s$ on $\{0,1\}^{s+1}$.

For a permutation $\mathbf{p}=(p_1,\dots,p_{s+1})$ of the set $\{1,\dots,s+1\}$, we take sequence of lines $(\ell_0^{\mathbf{p}},\dots,\ell_{s+1}^{\mathbf{p}})$ such that $\ell_0^{\mathbf{p}}=a$, $\ell_s^{\mathbf{p}}=b$, and $\ell_{i-1}^{\mathbf{p}}$, $\ell_{i}^{\mathbf{p}}$ satisfy the reflection law on $\mathcal{Q}_{i}$ for each $i\in\{1,\dots,s+1\}$.
Such a sequence exists and it is unique.
Moreover, if $k\in\{1,\dots,s+1\}$ is given, and permutations $\mathbf{p}$, $\mathbf{p}'$ coincide in first $k$ coordinates, then $\ell_{i}^{\mathbf{p}}=\ell_{i}^{\mathbf{p}'}$ for $i\le k$.
Then take $\{i_1,\dots,i_k\}$ to be a subset of $\{1,\dots,s+1\}$, and $\mathbf{p}$ any permutation of the set $\{1,\dots,s+1\}$ with $p_1=i_1$, \dots, $p_k=i_k$.
We define:
$$
\varphi_s(\chi_{\{i_1,\dots,i_k\}})=\ell_{k}^{\mathbf{p}},
$$
where $\chi$ is the corresponding characteristic function on $\{0,1\}^{s+1}$.

Subsequently, we define $\varphi_s$ on the rest of $\mathbf{Z}^{s+1}$, so that 
$\{\varphi_s(\mathbf{n}_0+i\mathbf{e}_j)\}_{i\in\mathbf{Z}}$ 
will represent billiard trajectories within $\mathcal{Q}_j$.

This construction is correct and unique due to the Theorem \ref{th:3d}.
\end{proof}

%% file: 5-DRnets/example-const.tex
Let $\mathcal{Q}_1$, \dots, $\mathcal{Q}_m$ be distinct quadrics belonging to a confocal family and $\ell$ a line in $\mathbf{P}^d$.
Let us choose lines $\ell_i$ satisfying with $\ell$ the reflection law on
$\mathcal{Q}_i$, $1\le i\le m$.

\begin{theorem}\label{th:const}
There is unique double reflection net $\varphi\ :\ \mathbf{Z}^m\to\mathcal{A}_{\ell}$, with the following properties:
\begin{itemize}
\item
$\varphi(0,\dots,0)=\ell$;

\item
$\varphi(\mathbf{e}_i)=\ell_i$, for each $i\in\{1,\dots,m\}$.
\end{itemize}
By $\mathcal{A}_{\ell}$, we denoted the set of all lines in $\mathbf{P}^d$ touching the same $d-1$ quadrics from the confocal family as $\ell$.
\end{theorem}
\begin{proof}
First, we define $\varphi$ on $\{0,1\}^m$, from the condition that lines corresponding to each $2$-face of the unit cube need to form a double reflection configuration.
This construction is unique because of Proposition \ref{prop:drc.quad} and correct, due to the $3D$-consistency property proved in Theorem \ref{th:3d}.

At all other points of $\mathbf{Z}^m$, $\varphi$ is uniquely defined from the request that $\{\varphi(\mathbf{n}_0+i\mathbf{e}_j)\}_{i\in\mathbf{Z}}$ will be billiard trajectories within $\mathcal{Q}_j$.

Consistency of the construction follows again from Theorem \ref{th:3d}.
\end{proof}

%% file: 5-DRnets/F.tex
Let
$\varphi\ :\ \mathbf{Z}^m \to \mathcal{A}$
be a double reflection net.

For given $\mathbf{n}_0\in\mathbf{Z}^m$ and distinct indices $i,j,k\in\{1,\dots m\}$, consider the following points of its $i$-th focal net:
\begin{gather*}
F_{i}=F^{(i)}(\mathbf{n}_0)=\varphi(\mathbf{n}_0)\cap\varphi(\mathbf{n}_0+\mathbf{e}_i),
   \\
F_{ij}=F^{(i)}(\mathbf{n}_0+\mathbf{e}_j)=
\varphi(\mathbf{n}_0+\mathbf{e}_j)\cap\varphi(\mathbf{n}_0+\mathbf{e}_j+\mathbf{e}_i),
   \\
F_{ik}=F^{(i)}(\mathbf{n}_0+\mathbf{e}_k)=
\varphi(\mathbf{n}_0+\mathbf{e}_k)\cap\varphi(\mathbf{n}_0+\mathbf{e}_k+\mathbf{e}_i),
	\\
F_{ijk}=F^{(i)}(\mathbf{n}_0+\mathbf{e}_j+\mathbf{e}_k)=
\varphi(\mathbf{n}_0+\mathbf{e}_j+\mathbf{e}_k)\cap\varphi(\mathbf{n}_0+\mathbf{e}_j+\mathbf{e}_k+\mathbf{e}_i).
\end{gather*}

\begin{proposition}\label{prop:coplanar}
Points 
$F_i$, $F_{ij}$, $F_{ik}$, $F_{ijk}$ are coplanar.
\end{proposition}

\begin{proof}
This is a consequence of the theorem of focal nets from \cite{BS2008book}.
However, we will show the direct proof, from configurations considered in
Section \ref{sec:billiard}.

The four points belong to quadric $\mathcal{Q}_{i}$.
Tangent planes to $\mathcal{Q}_{i}$ at these points, divided into two pairs, determine two pencils of planes.
According to Theorem \ref{th:zvezda}, the two pencils are coplanar; thus they intersect.
As a consequence, the lines of poles with respect to the quadric
$\mathcal{Q}_{i}$, which correspond to these two pencils of planes, also intersect.
Follows that the four points are coplanar.
\end{proof}

We are going to define an $F$-transformation of the double reflection net.

First, we select a quadric $\mathcal{Q}_{\delta}$ from the confocal family and introduce line $\ell'$ which satisfies with $\varphi(\mathbf{n}_0)$ the reflection law on
$\mathcal{Q}_{\delta}$.

By Theorem \ref{th:const}, it is possible to construct double reflection net
$\bar{\varphi}\ :\ \mathbf{Z}^{m+1}\to\mathcal{A}$,
such that:
\begin{itemize}
\item
$\bar{\varphi}(\mathbf{n},0)=\varphi(\mathbf{n})$;

\item
$\bar{\varphi}(\mathbf{n}_0,1)=\ell'$.
\end{itemize}

Now, we define:
$$
\varphi^{+}\ :\  \mathbf{Z}^m\to\mathcal{A}_{\ell},
\quad
\varphi^{+}(\mathbf{n})=\bar{\varphi}(\mathbf{n},1).
$$

\begin{proposition}
Map $\varphi^{+}$ is an $F$-transformation of $\varphi$. 
\end{proposition}

\begin{proof}
Lines $\varphi^{+}(\mathbf{n})$ and $\varphi(\mathbf{n})$ intersect, since they satisfy the reflection law on $\mathcal{Q}_{\delta}$.
\end{proof}

%% file: 5-DRnets/grassmannian.tex
Let us recall the definition of a Grassmannian Darboux net from \cite{ABS2009}:
a map from the edges of a regular square lattice $\mathbf{Z}^m$ to the Grassmannian
$\mathbf{G}^d_r$ of $r$-dimensional projective subspaces of $d$-dimensional projective space is \emph{a Grassmanian Darboux net} if the four $r$-spaces of an elementary quadrilateral belong to a $(2r+1)$-space.
For $r=0$, the ordinary Darboux nets from \cite{Schief2003} are obtained, where the four points of intersection associated to a quadrilateral, belong to a line.

Using the isomorphism between the Grassmannian $\mathbf{G}^d_r$  and
Grassmannian $\Gr(r+1, d+1)$ of $(r+1)$-dimensional vector subspaces
of the $(d+1)$-dimensional vector space and following the analytical
description of the Darboux nets from \cite{ABS2009}, we come to the
equations:
\begin{equation}
x_j^i=\rho^{ij}x^i+(I-\rho^{ij})x^j,
\end{equation}
where $x$ are $(d+1,r+1)$-matrices with appropriate normalization, and $\rho$ are invertible $(r+1,r+1)$-matrices with $I$ as the identity.
Developing further, we come to
$$
x^i_{jk}=\rho^{ij}_k(\rho^{ik}x^i+(I-\rho^{ik})x^k)+
(I-\rho^{ij}_k)(\rho^{jk}x^j+(I-\rho^{jk})x^k),
$$
and finally to a condition for a matrix-valued one form to be closed
$$
\rho^{ij}_k\rho^{ik}=\rho^{ik}_j\rho^{ij},
$$
see \cite{ABS2009}. Moreover, the closed one-form $\rho$ can be
represented as an exact one: $\rho^{ij}=s^i_j(s^i)^{-1}$ and the
rotation coefficients are $b^{ij}=((s^i)^{-1}-(s^i_j)^{-1})s^j$.

Now, consider a general double reflection net (\ref{eq:drnet}).

To each edge $(\mathbf{n_0},\mathbf{n_0}+\mathbf{e}_i)$ of $\mathbf{Z}^m$, we can associate plane which is tangent to $\mathcal{Q}_i$ at point 
$\varphi(\mathbf{n_0})\cap\varphi(\mathbf{n_0}+\mathbf{e}_i)$.
Since the lines corresponding to the vertices of a face form a double reflection configuration, the four planes associated to the edges belong to a pencil.

In this way, we see that a double reflection net induces map:
$$
E(\mathbf{Z}^m)\rightarrow \mathbf{G}^d_{d-1},
$$
where $E(\mathbf{Z}^m)$ is the set of all edges of the integer lattice $\mathbf{Z}^m$.

In this way, double reflection nets induce a subclass of dual Darboux nets.

It was shown in \cite{Schief2003} how to associate discrete integrable hierarchies to the Darboux nets.

%% file: 6-YB/yb.tex
\emph{Yang-Baxter map} is a map
$R:\mathcal{X}\times\mathcal{X}\to\mathcal{X}\times\mathcal{X}$,
satisfying the Yang-Baxter equation:
$$
R_{23}\circ R_{13}\circ R_{12}=R_{12}\circ R_{13}\circ R_{23},
$$
where $R_{ij}:\mathcal{X}\times\mathcal{X}\times\mathcal{X}\to\mathcal{X}\times\mathcal{X}\times\mathcal{X}$
acts as $R$ on the $i$-th and $j$-th factor in the product, and as the identity on the rest one, see \cite{ABS2004} and references therein.

Here, we are going to construct an example of Yang-Baxter map associated to confocal families of quadrics.
To begin, we fix a family of confocal quadrics in $\mathbf{CP}^{n}$:
\begin{equation}\label{eq:confocal.family}
\mathcal{Q}_{\lambda}\ :\ 
\frac{z_1^2}{a_1-\lambda}+\dots+\frac{z_d^2}{a_d-\lambda}
=
z_{n+1}^2,
\end{equation}
where $a_1$, \dots, $a_d$ are constants in $\mathbf{C}$, and $[z_1:z_2:\dots:z_{n+1}]$ are homogeneous coordinates in $\mathbf{CP}^{n}$.

Take $\mathcal{X}$ to be the space $\mathbf{CP}^{n*}$ dual to the $n$-dimensional projective space, i.e.\ the variety of all hyper-planes in $\mathbf{CP}^{n}$.
Note that a general hyper-plane in the space is tangent to exactly one quadric from family (\ref{eq:confocal.family}).
Besides, in a general pencil of hyper-planes, there are exactly two of them tangent to a fixed general quadric.

Now, consider a pair $x$, $y$ of hyper-planes.
They are touching respectively unique quadrics $\mathcal{Q}_{\alpha}$, $\mathcal{Q}_{\beta}$ from (\ref{eq:confocal.family}).
Besides, these two hyper-planes determine a pencil of hyper-planes.
This pencil contains unique hyper-planes $x'$, $y'$, other than $x$, $y$, that are tangent to $\mathcal{Q}_{\alpha}$, $\mathcal{Q}_{\beta}$ respectively.

We define
$R : \mathbf{CP}^{n*}\times\mathbf{CP}^{n*} \to \mathbf{CP}^{n*}\times\mathbf{CP}^{n*}$, in such a way that $R(x,y)=(x',y')$ if $(x',y')$ are obtained from $(x,y)$ in the just described way.

Maps
$$
R_{12},\ R_{13},\ R_{23}\ :\
\mathbf{CP}^{n*}\times\mathbf{CP}^{n*}\times\mathbf{CP}^{n*}
\to
\mathbf{CP}^{n*}\times\mathbf{CP}^{n*}\times\mathbf{CP}^{n*}
$$
are then defined as follows:
\begin{align*}
&R_{12}(x,y,z)=(x',y',z)\quad\text{for}\quad (x',y')=R(x,y);\\
&R_{13}(x,y,z)=(x',y,z')\quad\text{for}\quad (x',z')=R(x,z);\\
&R_{23}(x,y,z)=(x,y',z')\quad\text{for}\quad (y',z')=R(y,z).
\end{align*}

To prove the Yang-Baxter equation for map $R$, we will need the following

\begin{lemma}\label{lemma:YBE}
Let $\mathcal{Q}_{\alpha}$, $\mathcal{Q}_{\beta}$, $\mathcal{Q}_{\gamma}$
be three non-degenerate quadrics from family (\ref{eq:confocal.family}) and
$x$, $y$, $z$ respectively their tangent hyper-planes.
Take:
$$
(x_2,y_1)=R(x,y),\quad
(x_3,z_1)=R(x,z),\quad
(y_3,z_2)=R(y,z).
$$
Let $x_{23}$, $y_{13}$, $z_{12}$ be the joint hyper-planes of pencils determined by pairs $(x_3,y_3)$ and $(x_2,z_2)$, $(x_3,y_3)$ and $(y_1,z_1)$, $(y_1,z_1)$ and $(x_2,z_2)$ respectively.

Then $x_{23}$, $y_{13}$, $z_{12}$ are touching quadrics $\mathcal{Q}_{\alpha}$, $\mathcal{Q}_{\beta}$, $\mathcal{Q}_{\gamma}$ respectively.
\end{lemma}

\begin{proof}
This statement, formulated for the dual space in dimension $n=2$ is proved as \cite[Theorem 5]{ABS2004}.

Consider the dual situation in arbitrary dimension $n$.
The dual quadrics $\mathcal{Q}_{\alpha}^*$, $\mathcal{Q}_{\beta}^*$, $\mathcal{Q}_{\gamma}^*$ belong to a linear pencil and points $x^*$, $y^*$, $z^*$, dual to hyper-planes $x$, $y$, $z$, are respectively placed on these quadrics.
Take the two-dimensional plane contained these three points.
The intersection of the pencil of quadrics with this, and any other plane as well, represents a pencil of conics.
Thus, Theorem 5 from \cite{ABS2004} will remain true in any dimension.

This lemma is dual to this statement, thus the proof is complete. 
\end{proof}

\begin{theorem}\label{th:YBE}
Map $R$ satisfies the Yang-Baxter equation.
\end{theorem}

\begin{proof}
Let $x$, $y$, $z$ be hyper-planes in $\mathbf{CP}^{n}$.
We want to prove that
$$
R_{23}\circ R_{13}\circ R_{12}(x,y,z)=R_{12}\circ R_{13}\circ R_{23}(x,y,z).
$$

Denote by $\mathcal{Q}_{\alpha}$, $\mathcal{Q}_{\beta}$, $\mathcal{Q}_{\gamma}$ the quadrics from (\ref{eq:confocal.family}) touching $x$, $y$, $z$ respectively.

Let:
\begin{align*}
&(x,y,z)
 \xrightarrow{R_{12}}
(x_2,y_1,z)
 \xrightarrow{R_{13}}
(x_{23},y_1,z_1)
 \xrightarrow{R_{23}}
(x_{23},y_{13},z_{12}),
 \\
&(x,y,z)
 \xrightarrow{R_{23}}
(x,y_3,z_2)
 \xrightarrow{R_{13}}
(x_3,y_3,z_{12}')
 \xrightarrow{R_{12}}
(x_{23}',y_{13}',z_{12}').
\end{align*}
Now, apply Lemma \ref{lemma:YBE} to hyper-planes $x$, $y$, $z_2$.
Since:
$$
(x_2,y_1)=R(x,y),\quad
(x_3,z_{12}')=R(x,z_2),\quad
(y_3,z)=R(y,z_2),
$$
we have that joint hyper-plane of pencils $(x_3,y_3)$ and $(x_2,z)$ is touching $\mathcal{Q}_{\alpha}$ -- therefore, this plane must coincide with $x_{23}$ and $x_{23}'$, i.e.~$x_{23}=x_{23}'$.
Also, joint hyper-plane of pencils $(y_1,z_{12}')$ and $(x_2,z)$ is touching $\mathcal{Q}_{\gamma}$ -- therefore, this is $z_1$ and $z_{12}=z_{12}'$.
Finally, joint hyper-plane of pencils $(x_3,y_3)$ and $(y_1,z_{12}')$ is tangent to $\mathcal{Q}_{\beta}$ -- it follows this is $y_{13}=y_{13}'$, which completes the proof.
\end{proof}

\begin{remark}
Instead of defining $R$ to act on the whole space $\mathbf{CP}^{n*}\times\mathbf{CP}^{n*}$, we can restrict it on the product of two non-degenerate quadrics from (\ref{eq:confocal.family}), namely:
$$
R(\alpha,\beta)\ :\ \mathcal{Q}_{\alpha}^*\times\mathcal{Q}_{\beta}^*\to\mathcal{Q}_{\alpha}^*\times\mathcal{Q}_{\beta}^*,
$$
where pair $(x,y)$ of tangent hyper-planes is mapped into pair $(x_1,y_1)$ in such a way that $x$, $y$, $x_1$, $y_1$ belong to the same pencil.

The corresponding Yang-Baxter equation is:
$$
R_{23}(\beta,\gamma)\circ R_{13}(\alpha,\gamma)\circ R_{12}(\alpha,\beta)=
R_{12}(\alpha,\beta)\circ R_{13}(\alpha,\gamma)\circ R_{23}(\alpha,\beta),
$$
where both sides of the equation represent maps from
$\mathcal{Q}_{\alpha}^*\times\mathcal{Q}_{\beta}^*\times\mathcal{Q}_{\gamma}^*$ to itself.
\end{remark}

In \cite{ABS2004}, for irreducible algebraic varieties
$\mathcal{X}_1$, $\mathcal{X}_2$
\emph{a quadrirational mapping}
$F\ :\ \mathcal{X}_1\times\mathcal{X}_2$
is defined.
For such a map $F$ and any fixed pair 
$(x,y)\in\mathcal{X}_1\times\mathcal{X}_2$,
except from some closed subvarieties of codimension at least $1$,
the graph
$\Gamma_F \subset
\mathcal{X}_1\times\mathcal{X}_2\times\mathcal{X}_1\times\mathcal{X}_2$
intersects each of the sets
$\{x\}\times\{y\}\times\mathcal{X}_1\times\mathcal{X}_2$,
$\mathcal{X}_1\times\mathcal{X}_2\times\{x\}\times\{y\}$,
$\mathcal{X}_1\times\{y\}\times\{x\}\times\mathcal{X}_2$,
$\{x\}\times\mathcal{X}_2\times\mathcal{X}_1\times\{y\}$
exactly at one point (see \cite[Definition 3]{ABS2004}).
In other words, $\Gamma_{F}$ is a graph of four rational maps:
$F$, $F^{-1}$, $\bar{F}$, $\bar{F}^{-1}$.

The following Proposition is a generalization of \cite[Proposition 4]{ABS2004}.

\begin{proposition}
Map
$
R(\alpha,\beta)\ :\ \mathcal{Q}_{\alpha}^*\times\mathcal{Q}_{\beta}^*\to\mathcal{Q}_{\alpha}^*\times\mathcal{Q}_{\beta}^*,
$
is quadrirational.
It is an involution and it concides with its companion $\bar{R}(\alpha,\beta)$.
\end{proposition}